\documentclass[12pt,reqno]{amsart}

\usepackage{amssymb,amsthm,amsmath,mathrsfs,}
\usepackage{enumerate}
\usepackage{fullpage}
\usepackage{chngpage}

\numberwithin{equation}{section}

\begin{document}

\newtheorem{theorem}{Theorem}[section]
\newtheorem{lemma}[theorem]{Lemma}
\newtheorem{define}[theorem]{Definition}
\newtheorem{remark}[theorem]{Remark}
\newtheorem{corollary}[theorem]{Corollary}
\newtheorem{example}[theorem]{Example}
\newtheorem{assumption}[theorem]{Assumption}
\newtheorem{proposition}[theorem]{Proposition}
\newtheorem{conjecture}[theorem]{Conjecture}

\def\Ref#1{Ref.~\cite{#1}}

\def\Rnum{{\mathbb R}}
\def\const{\text{const.}}

\def\smallbinom#1#2{{\textstyle \binom{#1}{#2}}}

\def\pr{{\rm pr}}
\def\rk{{\rm rank}}
\def\spn{{\rm span}}
\def\dom{{\rm dom}\,}
\def\ran{{\rm ran}\,}
\def\hook{\rfloor\,}

\def\X{\mathbf{X}}
\def\Y{\mathbf{Y}}
\def\evw{\boldsymbol\omega}
\def\w{\boldsymbol\varpi}
\def\d{\mathrm{d}}
\def\id{\mathrm{id}}
\def\lieder#1{{\mathcal L}_{#1}}
\def\ad{\mathrm{ad}}
\def\t{\mathrm{t}}

\def\Jsp{\mathrm{J}}
\def\Esp{\mathcal{E}}
\def\symmsp{\mathrm{Symm}}
\def\adjsymmsp{\mathrm{AdjSymm}}

\def\Rop{{\mathcal R}}
\def\Dop{{\mathcal D}}
\def\Hop{{\mathcal H}}
\def\Jop{{\mathcal J}}
\def\lapl{\boldsymbol\Delta}
\def\grad{\boldsymbol\nabla}

\def\smallbinom#1#2{{\textstyle \binom{#1}{#2}}}

\def\const{\text{const.}}

\tolerance=50000
\allowdisplaybreaks[3]

\title{Geometrical formulation for adjoint-symmetries\\ of partial differential equations}

\author{
Stephen C. Anco${}^1$
\lowercase{\scshape{and}}
Bao Wang${}^2$
\\\\\scshape{
D\lowercase{\scshape{epartment}} \lowercase{\scshape{of}} M\lowercase{\scshape{athematics and}} S\lowercase{\scshape{tatistics}}\\
B\lowercase{\scshape{rock}} U\lowercase{\scshape{niversity}}\\
S\lowercase{\scshape{t.}} C\lowercase{\scshape{atharines}}, ON L2S3A1, C\lowercase{\scshape{anada}}
}}

\thanks{${}^1$sanco@brocku.ca, ${}^2$wangbao@lsec.cc.ac.cn}

\begin{abstract}
A geometrical formulation for adjoint-symmetries as 1-forms 
is studied for general partial differential equations (PDEs), 
which provides a dual counterpart of the geometrical meaning of symmetries
as tangent vector fields on the solution space of a PDE.
Two applications of this formulation are presented. 
Additionally,
for systems of evolution equations, 
adjoint-symmetries are shown to have another geometrical formulation
given by 1-forms that are invariant under the flow generated by the system on the solution space.
This result is generalized to systems of evolution equations with spatial constraints,
where adjoint-symmetry 1-forms are shown to be invariant up to a functional multiplier of a normal 1-form associated to the constraint equations.
All of the results are applicable to the PDE systems of interest 
in applied mathematics and mathematical physics. 
\end{abstract}

\maketitle

\section{Introduction}

Symmetries are a fundamental coordinate-free structure of
a partial differential equation (PDE). 
In geometrical terms,
an infinitesimal symmetry is an evolutionary (vertical) vector field
that is tangent to the solution space of a PDE,
where the components of the vector field are the solutions of
the linearization of the PDE on its solution space.
(See, e.g. \Ref{Ovs-book,Olv-book,KraVin-book,BCA-book}.)

Knowledge of the symmetries of a PDE can be used to
map given solutions into other solutions, 
find invariant solutions, 
detect and find mappings into a target class of PDEs,
detect integrability,
and find conservation laws through Noether's theorem
when a PDE has a variational (Lagrangian) structure.

Solutions of the adjoint linearization of a PDE on its solution space
are known as adjoint-symmetries.
This terminology was first introduced and explored for ordinary differential equations (ODEs) in \Ref{SarCanCra,SarBon,Sar,BA-book},
and then generalized to PDEs in \Ref{AncBlu1997,AncBlu2002b}.
(See \Ref{Anc-review} for a recent overview for PDEs.)
When a PDE lacks a variation structure, 
then its adjoint-symmetries will differ from its symmetries. 

Knowledge of the adjoint-symmetries of a PDE can be used for several purposes just as symmetries can. 
Specifically, 
solutions of the PDE can be found
analogously to the invariant surface condition associated to a symmetry;
mappings into a target class of PDEs can be detected and found
analogously to characterizing the symmetry structure of the target class;
integrability can be detected
analogously to existence of higher-order symmetries;
and conservation laws can be determined
analogously to symmetries that satisfy a variational condition. 
In particular,
the counterpart of variational symmetries for a general PDE
is provided by multipliers,
which are well known to be adjoint-symmetries that satisfy an Euler-Lagrange condition.

However, a simple geometrical meaning (apart from abstract formulations) for adjoint-symmetries
has yet to be developed in general for PDEs. 
Several significant new steps toward this goal will be taken in the present paper. 

Firstly,
for general PDE systems, 
adjoint-symmetries will be shown to correspond to evolutionary (vertical) 1-forms that functionally vanish on the solution space of the system.
This formulation has two interesting applications.
It will provide a geometrical derivation of a well-known formula that generates
a conservation law from a pair consisting of a symmetry and an adjoint-symmetry \cite{AncBlu1997,Anc2017}.
It also will yield three different actions of symmetries on adjoint-symmetries
from Cartan's formula for the Lie derivative, 
providing a geometrical formulation of some recent work that
used an algebraic viewpoint \cite{AncWan2020b}.

Secondly, 
for evolution systems, 
these adjoint-symmetry 1-forms will be shown to have the structure of a Lie derivative of a simpler underlying 1-form,
utilizing the flow generated by the system.
As a result,
adjoint-symmetries of evolution systems will geometrically correspond to 1-forms
that are invariant under the flow on the solution space of the system.
This directly generalizes the geometrical meaning of adjoint-symmetries known for ODEs \cite{Sar}.

Thirdly,
a bridge between the preceding results for general PDE systems and evolution systems
will be developed by considering evolution systems with spatial constraints.
These systems are ubiquitous in applied mathematics and mathematical physics,
for example:
Maxwell's equations,
incompressible fluid equations,
magnetohydrodynamical equations,
and Einstein's equations. 
For such systems,
invariance of the adjoint-symmetry 1-form under the constrained flow 
will be shown to hold up to a functional multiple of the normal 1-form associated to the constraint equations.

Throughout, the approach will be concrete, rather than abstract,
so that the results can be readily understood and applied to specific PDE systems of interest in applied mathematics and mathematical physics. 

The rest of the paper is organized as follows.
Section~\ref{sec:prelims} discusses
the evolutionary form of vector fields and its counterpart for 1-forms
in the mathematical framework of calculus in jet space,
which will underlie all of the main results. 
Section~\ref{sec:geomform} reviews the geometrical formulation of symmetries
and presents the counterpart geometrical formulation of adjoint-symmetries. 
In addition, some examples of adjoint-symmetries of physically interesting PDE systems 
are discussed. 
Section~\ref{sec:applications} gives the two applications of adjoint-symmetry 1-forms.
Section~\ref{sec:evolpde} develops the main results for adjoint-symmetries of evolution systems and extends these results to constrained evolution systems.
Some concluding remarks are made in Section~\ref{sec:remarks}.

\section{Vector fields, 1-form fields, and their evolutionary form}\label{sec:prelims}

To begin, some essential tools \cite{Olv-book,Nes,Anc-review}
from calculus in jet space will be reviewed. 
This will set the stage for a discussion of the evolutionary form of 
vector fields and its counterpart for 1-forms, 
as needed for the main results in the subsequent sections. 

Independent variables are denoted $x^i$, $i=1,\ldots,n$, 
and dependent variables are denoted $u^\alpha$, $\alpha=1,\ldots,m$. 
Derivative variables are indicated by subscripts employing a multi-index notation:
$I=\{i_1,\ldots,i_N\}$, $u^\alpha_I = u^\alpha_{i_1\cdots i_N} := \partial_{x^{i_1}}\cdots\partial_{x^{i_N}} u^\alpha$, $|I|=N$;
$I=\emptyset$, $u^\alpha_I := u^\alpha$, $|I|=0$. 
Some useful notation: 
$\partial^k u$ will denote the set $\{u^\alpha_I\}_{|I|=k}$ of all derivative variables of order $k\geq 0$; 
$u^{(k)}$ will denote the set $\{u^\alpha_I\}_{0\leq |I|\leq k}$ of all derivative variables of all orders up to $k\geq 0$. 
The summation convention of summing over any repeated (multi-) index in an expression is used throughout. 

Jet space is the coordinate space $\Jsp=(x^i,u^\alpha,u^\alpha_j,\ldots)$. 
A smooth function $u^\alpha=\phi^\alpha(x) :\Rnum^n\to \Rnum^m$
determines a point in $\Jsp$:
at any $x^i=(x_0)^i$, 
the values $(u_0)^\alpha:=\phi^\alpha(x_0)$ 
and the derivative values $(u_0)^\alpha_J:=\partial_{j_1}\cdots\partial_{j_N}\phi^\alpha(x_0)$ for all orders $N\geq 1$ 
give a map 
\begin{equation}\label{eval}
u^\alpha=\phi^\alpha(x)\stackrel{x_0}\to ((x_0)^i,(u_0)^\alpha,(u_0)^\alpha_j,\ldots) \in \Jsp . 
\end{equation}

In jet space, the primitive geometric objects consist of 
partial derivatives $\partial_{x^i}$, $\partial_{u^\alpha_J}$, 
and differentials $\d x^i$, $\d u^\alpha_J$. 
They are related by duality (hooking) relations:
\begin{align}
& \partial_{x^i} \hook \d x^j = \delta_i^j , 
\\
& \partial_{u^\alpha_I} \hook \d u^\beta_J = \delta_\alpha^\beta \delta_J^I . 
\end{align}
It will be useful to also introduce the geometric contact 1-forms 
\begin{equation}\label{contact1forms}
\Theta^\alpha_I = \d u^\alpha_I - u^\alpha_{Ii} \d x^i . 
\end{equation}
Under the evaluation map \eqref{eval}, 
the pull back of a contact 1-form vanishes. 

Total derivatives are given by $D_i= \partial_{x^i} + u^\alpha_{iJ}\partial_{u^\alpha_J}$,
which corresponds to the chain rule under the evaluation map \eqref{eval}. 
Higher total derivatives are defined by 
$D_J = D_{j_1}\cdots D_{j_N}$, $J=\{j_1,\ldots,j_N\}$, $|J|=N$. 
For $J=\emptyset$, $D_\emptyset = \id$ is the identity operator,
where $|\emptyset|=0$.
In particular, 
$D_J u^\alpha = u^\alpha_J$, and $D_J \d u^\alpha = \d u^\alpha_J$. 

A differential function is a function $f(x,u^{(k)})$ defined on a finite jet space 
$\Jsp^{(k)}=(x^i,u^\alpha,u^\alpha_j,\ldots,u^\alpha_{j_1\cdots j_k})$ of order $k\geq 0$. 
The Frechet derivative of a differential function $f$ is given by
\begin{align}
f' = f_{u^\alpha_I}D_I
\end{align}
which acts on (differential) functions $F^\alpha$. 
The adjoint-Frechet derivative of a differential function $f$ is given by
\begin{align}
(f'{}^*)_\alpha = (-1)^{|I|} D_I f_{u^\alpha_I}
\end{align}
which acts on (differential) functions $F$, 
where the righthand side is viewed as a composition of operators. 

The Frechet second-derivative is given by
\begin{align}
f''(F_1,F_2) = f_{u^\alpha_I u^\beta_J}(D_I F_1^\alpha)(D_J F_2^\beta) . 
\end{align}
This expression is symmetric in the pair of functions $(F_1^\alpha,F_2^\alpha)$.

The commutator of two differential functions $f_1$ and $f_2$ is given by $[f_1,f_2] = f_2'(f_1)- f_1'(f_2)$. 

The Euler operator (variational derivative) is given by 
\begin{equation}\label{Euler.op}
 E_{u^\alpha} = (-1)^{|I|} D_I\partial_{u^\alpha_I} . 
\end{equation}  
It characterizes total divergence expressions: 
$E_{u^\alpha}(f) =0$ holds identically iff $f= D_i F^i$
for some differential vector function $F^i(x,u^{(k)})$.
The product rule takes the form
\begin{equation}
E_{u^\alpha}(f_1f_2) = f_1'{}^*(f_2)_\alpha + f_2'{}^*(f_1)_\alpha . 
\end{equation}

The higher Euler operators 
\begin{equation}\label{higher.Euler.op}
E_{u^\alpha}^{I} = \smallbinom{I}{J} (-1)^{|J|} D_J\partial_{u^\alpha_{IJ}}
\end{equation}    
characterize higher-order total derivative expressions:
$E_{u^\alpha}^{I}(f) = 0$ holds identically iff $f= D_{i_1}\cdots D_{i_{|I|}} F^{i_1\ldots i_{|I|}}$
for some differential tensor function $F^{i_1\ldots i_{|I|}}(x,u^{(k)})$. 

The Frechet derivative is related to the Euler operator by 
\begin{equation}
f'(F) = F^\alpha E_{u^\alpha}(f) + D_i\Gamma^i(F;f),
\quad
\Gamma^i(F;f)= (D_J F^\alpha) E_{u^\alpha_{iJ}}(f) . 
\end{equation}  
The Frechet derivative and its adjoint are related by 
\begin{equation}
F_2 f'(F_1) - F_1^\alpha f'{}^*(F_2)_\alpha 
= D_i\Psi^i(F_1,F_2;f),
\quad
\Psi^i(F_1,F_2;f) = (D_K F_2) (D_J F_1^\alpha)E_{u^\alpha_{iJ}}^{K}(f) . 
\end{equation}

\subsection{Evolutionary vector fields and 1-form fields}

A \emph{vector field in jet space} is defined as the geometric object 
\begin{equation}
P^i\partial_{x^i} + P^\alpha_I \partial_{u^\alpha_I}
\end{equation}
whose components are differential functions. 
Similarly, a \emph{1-form field in jet space} is defined as the geometric object
\begin{equation}
Q_i \d x^i + Q_\alpha^I \d u^\alpha_I
\end{equation}
whose components are differential functions. 
Total derivatives $D_i = \partial_{x^i} + u^\alpha_{iI} \partial_{u^\alpha_I}$
represent trivial vector fields that annihilate contact 1-forms: 
$D_i\hook\Theta^\alpha_J = 0$.

Geometric counterparts of partial derivatives $\partial_{u^\alpha_J}$ are 
evolutionary (vertical) differentials $\d u^\alpha_J$, 
where $\d$ is the evolutionary version of $d$: $\d^2=0$, $\d x^i=0$. 
They satisfy the duality (hooking) relation:
\begin{align}
\partial_{u^\alpha_I} \hook \d u^\beta_J 
= \delta_\alpha^\beta \delta_J^I . 
\end{align}

An \emph{evolutionary (vertical) vector field} is the geometric object 
\begin{equation}\label{vert.vectorfield}
P^\alpha_I \partial_{u^\alpha_I}
\end{equation}
whose components are differential functions. 
Every vector field $\X = P^i\partial_{x^i} + P^\alpha_I \partial_{u^\alpha_I}$
has a unique evolutionary form 
$\hat\X = \X - P^i D_i = \hat P^\alpha_I \partial_{u^\alpha_I}$
given by the components $\hat P^\alpha_I = P^\alpha_I -P^i u^\alpha_{iI}$. 
Its dual counterpart is an \emph{evolutionary (vertical) 1-form field}
\begin{equation}\label{vert.1formfield}
Q_\alpha^I \d u^\alpha_I
\end{equation}
whose components are differential functions. 

For later developments, 
it will be useful to define the \emph{functional pairing relation}
\begin{align}\label{pairing}
\langle P^\alpha_I \partial_{u^\alpha_I}, Q_\alpha^I \d u^\alpha_I \rangle
= \int P^\alpha_I Q_\alpha^I\, dx
\end{align}
between evolutionary vector fields and evolutionary 1-form fields. 
In local form, this pairing is given by the expression
\begin{align}\label{pairing.local}
P^\alpha_I Q_\alpha^I \text{ mod total } D . 
\end{align}

Two evolutionary 1-forms will be considered \emph{functionally equivalent} 
iff their pairings with an arbitrary evolutionary vector field agree:
\begin{equation}
\langle P^\alpha_I \partial_{u^\alpha_I}, Q_1{}_\beta^J \d u^\beta_J \rangle
= \langle P^\alpha_I \partial_{u^\alpha_I}, Q_2{}_\beta^J \d u^\beta_J \rangle , 
\end{equation}
or in local form 
\begin{align}
P^\alpha_I (Q_1{}_\alpha^I -Q_2{}_\alpha^I) =0 \text{ mod total } D . 
\end{align}

Functional equivalence of 1-forms is closely related to the notion of
functional 1-forms in the variational bi-complex.
See \Ref{Olv-book} for details.

\section{Geometric formulation of symmetries and adjoint-symmetries}\label{sec:geomform}

Consider a general PDE system of order $N$ consisting of $M$ equations
\begin{equation}\label{pde}
G^A(x,u^{(N)}) =0,
\quad
A=1,\ldots,M
\end{equation}
where $x^i$, $i=1,\ldots,n$, are the independent variables,
and $u^\alpha$, $\alpha=1,\ldots,m$, are the dependent variables.
The space of formal solutions $u^\alpha(x)$ of the PDE system will be denoted $\Esp$.

There are many equivalent starting points for the formulation of infinitesimal symmetries. 
For the present purpose, 
the most useful one is given by evolutionary vector fields 
and utilizes only the Frechet derivative. 
A \emph{symmetry} is a vector field 
\begin{equation}\label{P.vector}
\X_P = P^\alpha(x,u^{(k)})\partial_{u^\alpha}
\end{equation}
whose component functions $P^\alpha(x,u^{(k)})$ are non-singular on $\Esp$ 
and satisfy the linearization of the PDE system on $\Esp$:
\begin{equation}\label{symm.deteqn}
(\pr \X_P G^A)|_\Esp = G'(P)^A|_\Esp =0 . 
\end{equation}
This is the \emph{symmetry determining equation}, 
and the functions $P^\alpha$ are called the characteristic of the symmetry.

In this setting, 
an \emph{adjoint-symmetry} consists of functions $Q_A(x,u^{(l)})$ 
that are non-singular on $\Esp$ and that satisfy the adjoint linearization of the PDE system on $\Esp$:
\begin{equation}\label{adjsymm.deteqn}
G'{}^*(Q)_\alpha|_\Esp =0 .
\end{equation}
This is the \emph{adjoint-symmetry determining equation}. 

In particular, 
the two determining equations \eqref{symm.deteqn} and \eqref{adjsymm.deteqn} 
are formal adjoints of each other. 
They coincide only in two cases: 
either $G'=G'{}^*$, 
which is the necessary and sufficient condition for a PDE system to be an Euler-Lagrange equation 
(namely, possess a variational structure)
\cite{Olv-book,BCA-book,Anc-review}; 
or $G'=-G'{}^*$, 
which is the necessary and sufficient condition for a PDE system to be a linear, constant-coefficient system of odd order \cite{AncBlu2002b}.

Since $P^\alpha$ has the geometrical status as the components of the vector field \eqref{P.vector}, 
a natural question is whether $Q_A$ has any status given by 
the components of some other geometrical object \cite{Anc2017,Anc-review}. 

It will be useful to work with a coordinate-free description of the PDE system \eqref{pde} in jet space. 
Such a system of equations $(G^1(x,u^{(N)}),\ldots,G^M(x,u^{(N)})) =0$ 
describes a set of $M$ surfaces in the finite space 
$\Jsp^{(N)}(x,u,\partial u,\ldots,\partial^N u)$. 
Total derivatives of these equations, 
$(D_I G^1(x,u^{(N)}),\ldots, D_I G^M(x,u^{(N)})) =0$, 
correspondingly describe sets of surfaces in the higher-derivative finite spaces
$\Jsp^{(N+|I|)}(x,u,\partial u,\ldots,\partial^{N+|I|} u)$. 
Altogether, the set comprised by the equations and the derivative equations for all orders $|I|\geq 0$
corresponds to an infinite set of surfaces in jet space,
which can be identified with the solution space $\Esp$. 

As is well known, symmetry vector fields geometrically 
describe tangent vector fields with respect to $\Esp$. 
To see this explicitly, 
first consider the identities
\begin{gather}
\d G^A = (G^A)_{u^\alpha_I} \d u^\alpha_I , 
\label{nor.1form}
\\
G'(P)^A = \pr\X_P G^A = \pr\X_P\hook \d G^A . 
\label{Frechet.hook.rel}
\end{gather}
Now observe that $\d G^A$ is the normal 1-form to the surfaces $G^A=0$. 
The symmetry determining equation \eqref{symm.deteqn} then shows that 
the prolonged vector field $\pr\X_P$ is annihilated by the normal 1-form
and hence is tangent to these surfaces 
iff $\X_P$ is a symmetry of the PDE system. 

This normal 1-form \eqref{nor.1form} provides a natural way to associate 
a 1-form to an adjoint-symmetry via 
\begin{equation}\label{Q.1form}
\w_Q = Q_A(x,u^{(l)}) \d G^A . 
\end{equation}
A functionally equivalent 1-form is obtained through integration by parts:
\begin{equation}
Q_A \d G^A  = Q_A (G^A)'(\d u) = G'{}^*(Q)_\alpha \d u^\alpha \text{ mod total } D . 
\end{equation}
Evaluating this 1-form on the solution space $\Esp$ then gives
\begin{equation}
\w_Q|_\Esp = 0 \text{ mod total } D . 
\end{equation}
Thus, a 1-form $\w_Q$ functionally vanishes on the surfaces $\Esp$ 
iff its components $Q_A$ are an adjoint-symmetry. 

This establishes a main geometrical result. 

\begin{theorem}\label{thm:adjsymm}
Adjoint-symmetries describe evolutionary 1-forms $Q_A\d G^A$ 
that functionally vanish on the solution space $\Esp$ of a PDE system \eqref{pde}. 
\end{theorem}

These developments have used evolutionary (vertical) vector fields and evolutionary 1-forms.
It is straightforward to reformulate everything in terms of 
full vector fields and full 1-forms. 

First, consider the normal 1-form 
\begin{equation}
\begin{aligned}
\d G^A & = (G^A)_{x^i} \d x^i + (G^A)'(\d u) \\
& = (G^A)'(\Theta) + ( (G^A)_{x^i}+ (G^A)'(u_i)  ) \d x^i \\
& = (G^A)'(\Theta) + D_i G^A \d x^i
\end{aligned}
\end{equation}
which yields the relation 
\begin{equation}
\d G^A|_\Esp = (G^A)'(\Theta)|_\Esp . 
\end{equation}
Then, observe 
\begin{equation}
\begin{aligned}
Q_A \d G^A |_\Esp
& = Q_A (G^A)'(\Theta)|_\Esp \\
& = (G^A)'{}^*(Q_A)_\alpha|_\Esp \Theta^\alpha \text{ mod total } D . 
\end{aligned}
\end{equation}
As a consequence, 
$Q_A \d G^A |_\Esp$ vanishes mod total $D$ 
iff $Q_A$ satisfies the adjoint-symmetry determining equation \eqref{adjsymm.deteqn}.
Moreover, the determining equation itself can be expressed directly in terms of 
the 1-form $Q_A \d G^A |_\Esp$ by 
$E_{\Theta^\alpha}(Q_A \d G^A)|_\Esp = (G^A)'{}^*(Q_A)|_\Esp =0$. 

\begin{proposition}
The adjoint-symmetry determining equation \eqref{adjsymm.deteqn}
can be expressed geometrically as 
\begin{equation}
E_{\Theta^\alpha}(Q_A \d G^A)|_\Esp = 0 .
\end{equation}
\end{proposition}

\subsection{Examples of adjoint-symmetries}

To illustrate the results, 
some examples of PDEs that possess non-trivial adjoint-symmetries will be given. 

The Korteweg-de Vries (KdV) equation 
\begin{equation}
u_t + uu_x + u_{xxx}=0
\end{equation}
for shallow water waves 
is an example of an evolutionary wave equation. 
Its symmetries $\X=P\partial_u$ are the solutions of the determining equation 
\begin{equation}
G'(P)|_\Esp =(D_t P +D_x(uP) +D_x^3 P)|_\Esp =0,
\end{equation}
with $G'= D_t + D_x u + D_x^3$ being the Frechet derivative of the KdV equation,
where $P$ is a non-singular function of $t$, $x$, $u$, and derivatives of $u$ 
on the space of KdV solutions $\Esp$. 
The determining equation for adjoint-symmetries $\w = Q G'(\d u)$
is the adjoint equation 
\begin{equation}
G'{}^*(Q)|_\Esp = (-D_t Q -uD_x Q - D_x^3 Q)|_\Esp =0, 
\end{equation}
where $Q$ is a non-singular function of $t$, $x$, $u$, and derivatives of $u$ on $\Esp$. 

KdV adjoint-symmetries up to first-order $Q(t,x,u,u_t,u_x)$ 
are given by \cite{AncBlu1997} the span of 
\begin{equation}
Q^{(1)}= 1,
\quad
Q^{(2)}=u,
\quad
Q^{(3)}=tu-x . 
\end{equation}
The first two are part of a hierarchy of higher-order adjoint-symmetries generated by
a recursion operator $\mathcal{R} = D_x^2 + \tfrac{1}{3}u + \tfrac{1}{3}D_x^{-1}uD_x$ 
applied to $Q=u$. 
The third one along with all of the ones in the hierarchy
are related to symmetries of the KdV equation through the Hamiltonian operator
$\mathcal{H} = D_x$. 
If a linear combination of the lowest-order adjoint-symmetries is used 
like an invariant surface condition, $c_1 + c_2 (tu-x) + c_3 u =0$, 
then this yields $u = (c_2 x -c_1)/(c_2 t + c_3)$ 
which is a similarity solution of the KdV equation. 

An example of a non-evolutionary equation is 
\begin{equation}
\Delta \phi_t + \phi_x \Delta \phi_y - \phi_y\Delta \phi_x =0
\end{equation}
which governs the vorticity $\Omega=\Delta \phi$ for incompressible inviscid fluid flow 
in two spatial dimensions,
where the fluid velocity has the components $\vec v=(-\phi_y,\phi_x)$. 
The symmetries $\X = P\partial_\phi$ of this equation 
are the solutions of the determining equation 
\begin{equation}
G'(P)|_\Esp = (D_t \lapl P +\phi_x D_y \lapl P +\Delta \phi_y D_x P -\phi_y D_x \lapl P -\Delta \phi_x D_y P)|_\Esp =0, 
\end{equation}
where $P$ is a non-singular function of $t$, $x$, $y$, $\phi$, and derivatives of $\phi$ 
on the space of vorticity solutions $\Esp$,
with $G'=D_t \lapl +\phi_x D_y \lapl +\Delta \phi_y D_x  -\phi_y D_x \lapl  -\Delta \phi_x D_y$
being the Frechet derivative of the vorticity equation
given in terms of the total Laplacian operator $\lapl = D_x^2+D_y^2$. 
The determining equation for adjoint-symmetries $\w = Q G'(\d \phi)$
is the adjoint equation 
\begin{equation}
G'{}^*(Q)|_\Esp = -(D_t \lapl Q +D_y \lapl(\phi_x Q) +D_x( \Delta \phi_y Q) - D_x \lapl (\phi_yQ) -D_y(\Delta \phi_x Q)|_\Esp =0, 
\end{equation}
where $Q$ is a non-singular function of $t$, $x$, $y$, $\phi$, and derivatives of $\phi$ on $\Esp$. 

The first-order adjoint-symmetries $Q(t,x,y,\phi,\phi_t,\phi_x,\phi_y)$ 
are given by \cite{AncWan2020b} the span of 
\begin{equation}
Q^{(1)}=x^2+y^2,
\quad
Q^{(2)}=\phi,
\quad
Q^{(3)}=f(t),
\quad
Q^{(4)}=xf(t),
\quad
Q^{(5)}=yf(t),
\end{equation}
where $f(t)$ is an arbitrary smooth function. 
If a linear combination of these adjoint-symmetries is used 
like an invariant surface condition, 
$c_1(x^2+y^2) + c_2 \phi + c_3 f(t) +c_4 xf(t) + c_5 yf(t) =0$, 
then taking $c_2=-1$ gives $\phi = c_1(x^2+y^2) + (c_3 +c_4 x + c_5 y) f(t)$
which is a constant vorticity solution, 
with $\Omega= 2c_1$ and $\vec v = (-2c_1 y+c_5 f(t),2c_1 x+c_4 f(t))$. 

Maxwell's equations in free space 
are an example of an evolution system with spatial constraints:
\begin{equation}
{\vec E_t -\nabla\times \vec B} = 0,
\quad
{\vec B_t +\nabla\times \vec E} = 0,
\quad
{\nabla\cdot\vec E}={\nabla\cdot\vec B}=0
\end{equation}
(in relativistic units with the speed of light set to $1$).
The symmetries $\X = \vec P^E\cdot\partial_{\vec E} + \vec P^B\cdot\partial_{\vec B}$ of this system 
are the solutions of the determining equations 
\begin{equation}
G'\begin{pmatrix} \vec P^E \\ \vec P^B \end{pmatrix}\bigg|_\Esp = 
\begin{pmatrix} 
(D_t \vec P^E -\grad\times \vec P^B)|_\Esp \\
(D_t \vec P^B +\grad\times \vec P^E)|_\Esp \\
(\grad\cdot \vec P^E)|_\Esp \\ 
(\grad\cdot \vec P^B)|_\Esp \\
\end{pmatrix}
=0, 
\end{equation}
where $\vec P^E$ and $\vec P^B$ are non-singular vector functions of 
$t$, $x$, $y$, $z$, $\vec E$, $\vec B$, and derivatives of $\vec E$, $\vec B$
on the space of Maxwell solutions $\Esp$,
with 
$G'=\begin{pmatrix} D_t & -\grad\times \\ \grad\times & D_t \\ \grad\cdot & 0 \\ 0 & \grad\cdot \end{pmatrix}$
being the Frechet derivative of the system
in terms of the total derivative operator $\grad = (D_x,D_y,D_z)$. 
The determining equation for adjoint-symmetries 
$\w =\begin{pmatrix} \vec Q^E & \vec Q^B & Q^E & Q^B \end{pmatrix} G'\begin{pmatrix} \d\vec E \\ \d\vec B \end{pmatrix}$
is the adjoint equation 
\begin{equation}
G'{}^*\begin{pmatrix} \vec Q^E & \vec Q^B & Q^E & Q^B \end{pmatrix}\Big|_\Esp = 
\begin{pmatrix} 
(-D_t \vec Q^E +\grad\times \vec Q^B -\grad Q^E)|_\Esp \\ 
(-D_t \vec Q^B -\grad\times \vec Q^E - \grad Q^B)|_\Esp 
\end{pmatrix}
=0, 
\end{equation}
where the vectors $\vec Q^E$, $\vec Q^B$, and the scalars $Q^E$, $Q^B$,
are non-singular functions of $t$, $x$, $y$, $z$, $\vec E$, $\vec B$, and derivatives of $\vec E$, $\vec B$ on $\Esp$. 
Note that the adjoint $*$ here includes a matrix transpose applied to the row matrix 
comprising the adjoint-symmetry vector and scalar functions. 

Because Maxwell's equations are a linear system and contain constraints, 
it possess three types of adjoint-symmetries \cite{AncPoh2002,AncPoh2004}: 
elementary adjoint-symmetries such that $\vec Q^E$, $\vec Q^B$, $Q^E$, $Q^B$ are functions only of $t$, $x$, $y$, $z$;
gauge adjoint-symmetries given by 
$\vec Q^E = \grad\chi^E$, $\vec Q^B=\grad\chi^B$, $Q^E=-D_t\chi^E$, $Q^B=-D_t\chi^B$ 
in terms of scalars $\chi^E$ and $\chi^B$ that are arbitrary non-singular functions of 
$t$, $x$, $y$, $z$, $\vec E$, $\vec B$, and derivatives of $\vec E$, $\vec B$ on $\Esp$;
and a hierarchy of linear adjoint-symmetries. 
The linear adjoint-symmetries of zeroth order are given by the span of 
\begin{equation}
\vec Q^E = \vec\xi\times\vec B +\zeta\vec E,
\quad
\vec Q^B = -\vec\xi\times\vec E +\zeta\vec B,
\quad
Q^E = \vec\xi\cdot\vec E,
\quad
Q^B = \vec\xi\cdot\vec B
\end{equation}
and 
\begin{equation}
\vec Q^E = \vec\xi\times\vec E -\zeta\vec B,
\quad
\vec Q^B = \vec\xi\times\vec B +\zeta\vec E,
\quad
Q^E = -\vec\xi\cdot\vec B,
\quad
Q^B = \vec\xi\cdot\vec E
\end{equation}
where 
\begin{equation}
\begin{aligned}
\vec\xi & = \vec a_0 +\vec a_1 \times \vec x + \vec a_2 t + a_3 \vec x + a_4 t \vec x + (\vec a_5 \cdot \vec x) \vec x - \tfrac{1}{2}\vec a_5 (\vec x\cdot\vec x+ t^2) ,
\\
\zeta & = a_0 + \vec a_2\cdot\vec x +a_3 t +\tfrac{1}{2} a_4 (\vec x\cdot\vec x+ t^2) +(\vec a_5\cdot\vec x)t ,
\end{aligned}
\end{equation}
in terms of arbitrary constant scalars $a_0$, $a_3$, $a_4$, 
and arbitrary constant vectors $\vec a_0$, $\vec a_1$, $\vec a_2$, $\vec a_5$, 
with $\vec x = (x,y,z)$. 
The pair $(\vec\xi,\zeta)$ represents a conformal Killing vector in Minkowski space $\Rnum^{3,1}$. 

These two zeroth-order adjoint-symmetries are related by the duality symmetry 
$(\vec E,\vec B)\to (\vec B,-\vec E)$. 
The linear first-order adjoint-symmetries are more complicated and involve conformal Killing-Yano tensors. 
All higher-order adjoint-symmetries can be obtained from the zeroth and first order adjoint-symmetries by taking Lie derivatives with respect to conformal Killing vectors. 
Their explicit description can be found in \Ref{AncPoh2002,AncPoh2004}. 
An unexplored question is whether the lowest-order adjoint-symmetries 
can be used like an invariant surface condition to produce solutions of Maxwell's equations.

\section{Some applications}\label{sec:applications}

Two geometrical applications of Theorem~\ref{thm:adjsymm}
will be presented. 
The first application is a geometrical derivation of a well-known formula that generates
a conservation law from a pair consisting of a symmetry and an adjoint-symmetry. 
This derivation will use the functional pairing \eqref{pairing}. 
The second application is a geometrical derivation of three actions of symmetries on adjoint-symmetries. 
These symmetries actions have been obtained in recent work using an algebraic point of view \cite{AncWan2020b}.
They will be shown here to arise from Cartan's formula for the Lie derivative of
an adjoint-symmetry 1-form \eqref{Q.1form}. 

It will be useful to work with the determining equations 
for symmetries and adjoint-symmetries off of the solution space $\Esp$ 
of a given PDE system \eqref{pde}.
More precisely, the determining equations will be expressed in the full jet space containing $\Esp$. 

\begin{remark}
A PDE system \eqref{pde} will be assumed to be regular \cite{Anc-review}, 
so that Hadamard's lemma holds: 
a differential function $f$ satisfies $f|_\Esp =0$ iff $f=R_f(G)$,
where $R_f$ is a linear differential operator whose coefficients are non-singular on $\Esp$. 
\end{remark}

Consequently, 
for symmetries, 
$G'(P)^A|_\Esp=0$ holds iff 
\begin{equation}\label{symm.deteqn.offsoln} 
G'(P)^A = R_P(G)^A , 
\end{equation}
and likewise for adjoint-symmetries, 
$G'{}^*(Q)_\alpha|_\Esp=0$ holds iff 
\begin{equation}\label{adjsymm.deteqn.offsoln}
G'{}^*(Q)_\alpha = R_Q(G)_\alpha , 
\end{equation}
where $R_P$ and $R_Q$ are linear differential operators whose coefficients are non-singular on $\Esp$.

\subsection{Conservation laws from symmetries and adjoint-symmetries}\label{sec:conslaw}

The functional pairing \eqref{pairing} between 
a symmetry vector field \eqref{P.vector}
and an adjoint-symmetry 1-form \eqref{Q.1form} is given by 
\begin{align}\label{P.Q.pairing}
\langle \pr\X_P, \w_Q \rangle 
= \langle \pr P^\alpha \partial_{u^\alpha}, Q_A \d G^A \rangle
= \int Q_A G'(P)^A\, dx 
\end{align}
from identity \eqref{Frechet.hook.rel}. 
This pairing in local form \eqref{pairing.local} is the expression 
\begin{equation}
Q_A G'(P)^A \text{ mod total } D . 
\end{equation}
There are two different ways to evaluate it. 

First, since $\X_P$ is a symmetry, 
$Q_A G'(P)^A = Q_A R_P(G)^A$. 
Second, since $\w_Q$ is an adjoint-symmetry, 
$Q_A G'(P)^A = G'{}^*(Q)_\alpha P^\alpha +D_i\Psi^i(P,Q)_G 
= P^\alpha R_Q(G)_\alpha +D_i\Psi^i(P,Q;G)$,
where
\begin{equation}\label{current.P.Q}
\Psi^i(P,Q;G) = (D_K Q_A) (D_J P^\alpha)E_{u^\alpha_{iJ}}^{K}(G^A) . 
\end{equation}
Hence, on $\Esp$, 
$Q_A G'(P)^A|_\Esp = D_i\Psi^i(P,Q)_G|_\Esp = 0$
which is equivalent to $\langle \pr\X_P, \w_Q \rangle|_\Esp =0$.
This establishes the following conservation law. 

\begin{theorem}\label{thm:pairing.vanish}
Vanishing of the functional pairing \eqref{P.Q.pairing}
for any symmetry \eqref{P.vector} and any adjoint-symmetry \eqref{Q.1form}
corresponds to a conservation law
\begin{equation}
D_i\Psi^i(P,Q;G)|_\Esp = 0
\end{equation}
holding for the PDE system $G^A=0$,
where the conserved current $\Psi^i(P,Q;G)$ is given by expression \eqref{current.P.Q}.
\end{theorem}

\subsection{Action of symmetries on adjoint-symmetries}\label{sec:symmaction}

For any PDE system \eqref{pde}, 
its set of adjoint-symmetries is a linear space,
and as shown in \Ref{AncWan2020b}, 
symmetries of the PDE system have three different actions on this space. 

The primary symmetry action can be derived from 
the Lie derivative of an adjoint-symmetry 1-form 
with respect to a symmetry vector field. 

\begin{proposition}\label{prop:lieder.action}
If $\w_Q$ is an adjoint-symmetry 1-form \eqref{Q.1form}, 
namely $\w_Q|_\Esp =0 \text{ (mod total $D$)}$, 
then its Lie derivative with respect to any symmetry vector $\X_P=P^\alpha\partial_{u^\alpha}$
yields an adjoint-symmetry 1-form, 
\begin{equation}
\lieder{\X_P}\w_Q|_\Esp = \w_{S_P(Q)}|_\Esp = 0 \text{ (mod total $D$) }
\end{equation}
where
\begin{equation}\label{lieder.symmaction}
S_P(Q)_A = Q'(P)_A + R_P^*(Q)_A
\end{equation}
are its components. 
\end{proposition}

Here and throughout, 
$R_P$ and $R_Q$ are the linear differential operators
determined by equations \eqref{symm.deteqn.offsoln} and \eqref{adjsymm.deteqn.offsoln}.
The adjoints of these operators are denoted $R_P^*$ and $R_Q^*$. 

\begin{proof}

Recall that the Lie derivative has the following properties:
it acts as a derivation;
it commutes with the differential $d$;
it reduces to the Frechet derivative when acting on a differential function. 

By use of these properties, 
\begin{equation}
\begin{aligned}
\lieder{\X_P}\w_Q 
& = \lieder{\X_P}(Q_A \d G^A) \\
& = (\lieder{\X_P}Q_A) \d G^A + Q_A \lieder{\X_P}(\d G^A) \\
& = Q'(P)_A \d G^A + Q_A \d(G'(P)^A) \\
& = Q'(P)_A \d G^A + Q_A \d(R_P(G)^A) . 
\end{aligned}
\end{equation}
The last term can be simplified on $\Esp$:
$Q_A \d(R_P(G)^A)|_\Esp= Q_A R_P(\d G)^A|_\Esp = R_P^*(Q)_A \d G^A$ 
$\text{(mod total $D$)}$.
This yields 
\begin{align}
\lieder{\X_P}\w_Q |_\Esp
= (( Q'(P)_A + R_P^*(Q)_A ) \d G^A)|_\Esp \text{ (mod total $D$) }, 
\end{align}
completing the derivation. 

\end{proof}

There is an elegant formula, due to Cartan, for the Lie derivative in terms of 
the operations $d$ and $\hook$. 
This formula gives rise to two additional symmetry actions. 

\begin{theorem}\label{Cartan.symm.action}
The terms in Cartan's formula
\begin{equation}\label{cartan.formula} 
\lieder{\X_P}\w_Q = \d(\pr\X_P \hook \w_Q) +\pr\X_P \hook (d\w_Q)
\end{equation}
evaluated on $\Esp$
each yield an action of symmetries on adjoint symmetries. 
The action produced by the Lie derivative term has the components \eqref{lieder.symmaction},
and the actions produced by the differential term and the hook term 
respectively have the components
\begin{align}
S_{1\,P}(Q) & = R_P^*(Q)_A - R_Q^*(P)_A , 
\label{1stterm.symmaction}
\\
S_{2\,P}(Q) & = Q'(P)_A + R_Q^*(P)_A . 
\label{2ndterm.symmaction}
\end{align}
\end{theorem}

\begin{proof}

Consider the first term on righthand side in the formula \eqref{cartan.formula}. 
It can be evaluated in two different ways. 
Firstly, 
$\pr\X_P\hook (Q_A \d G^A) = Q_A G'(P)^A = Q_A R_P(G)^A$
yields
\begin{equation}\label{1stterm.simp1}
\d(\pr\X_P\hook (Q_A \d G^A))|_\Esp 
= \d(Q_A R_P(G)^A)|_\Esp
= (Q_A R_P(\d G^A ))|_\Esp
= (R_P^*(Q)_A \d G^A)|_\Esp . 
\end{equation}
Secondly, 
$Q_A \d G^A = R_Q(G)_\alpha\Theta^\alpha + Q_A (D_i G^A) \d x^i \text{ (mod total $D$) }$
gives
$\pr\X_P\hook (Q_A \d G^A) 
= \pr\X_P\hook ( R_Q(G)_\alpha\Theta^\alpha + Q_A (D_i G^A) \d x^i  \text{ (mod total $D$) } )
= R_Q(G)_\alpha P^\alpha \text{ (mod total $D$)}$. 
This yields 
\begin{equation}\label{1stterm.simp2}
\begin{aligned}
\d(\pr\X_P\hook (Q_A \d G^A))|_\Esp 
& = \d( R_Q(G)_\alpha P^\alpha \text{ (mod total $D$) } )|_\Esp\\
& = ( R_Q(\d G)_\alpha P^\alpha \text{ (mod total $D$) } )|_\Esp\\
& = ( R_Q^*(P)_A \d G^A \text{ (mod total $D$) } )|_\Esp . 
\end{aligned}
\end{equation}

Then, 
equating expressions \eqref{1stterm.simp1} and \eqref{1stterm.simp2} leads to 
the result 
\begin{equation}
( (R_P^*(Q)_A - R_Q^*(P)_A) \d G^A )|_\Esp = 0 \text{ (mod total $D$)}|_\Esp . 
\end{equation}
This equation shows that the symmetry action \eqref{1stterm.symmaction}
produces an adjoint-symmetry.

Now consider second term on righthand side in the formula \eqref{cartan.formula}. 
Similarly to the first term, 
it can be evaluated in two different ways. 
Firstly, 
$\d\w_Q = \d Q_A\wedge \d G^A$
yields
\begin{equation}
\pr\X_P\hook( \d Q_A\wedge \d G^A )
= Q'(P)_A \d G^A - G'(P)^A \d Q_A
= Q'(P)_A \d G^A - R_P(G)^A \d Q_A . 
\end{equation}
Hence, on $\Esp$, 
\begin{equation}\label{2ndterm.simp1}
(\pr\X_P\hook( \d Q_A\wedge \d G^A ))|_\Esp
= ( Q'(P)_A \d G^A )|_\Esp . 
\end{equation}
Secondly, 
$\d\w_Q = \d( R_Q(G)_\alpha \Theta^\alpha + Q_A (D_i G^A) \d x^i ) \text{ (mod total $D$)}$
gives
\begin{equation}
\d\w_Q|_\Esp = ( R_Q(\d G)_\alpha\wedge \Theta^\alpha + Q_A (D_i \d G^A)\wedge \d x^i )|_\Esp
\text{ (mod total $D$) }. 
\end{equation}
This yields 
\begin{equation}\label{2ndterm.simp2}
\begin{aligned}
& (\pr\X_P\hook( R_Q(\d G)_\alpha\wedge \Theta^\alpha + Q_A (D_i \d G^A)\wedge \d x^i ))|_\Esp\\
& =( R_Q(G'(P))_\alpha \Theta^\alpha - P^\alpha R_Q(\d G)_\alpha + Q_A (D_i G'(P)^A) \d x^i )|_\Esp\\
& = -( R_Q^*(P)_A \d G^A )|_\Esp \text{ (mod total $D$)} . 
\end{aligned}
\end{equation}
Equating expressions \eqref{2ndterm.simp1} and \eqref{2ndterm.simp2}
then gives the equation
\begin{equation}
( (Q'(P)_A + R_Q^*(P)_A) \d G^A )|_\Esp = 0 \text{ (mod total $D$)}|_\Esp 
\end{equation}
showing that the symmetry action \eqref{2ndterm.symmaction}
produces an adjoint-symmetry. 

\end{proof}

Observe that the three actions \eqref{lieder.symmaction}, 
\eqref{1stterm.symmaction}, \eqref{2ndterm.symmaction}
are related by 
\begin{equation}
S_{1\,P}(Q) +S_{2\,P}(Q) = S_{P}(Q) . 
\end{equation}
Each action is mapping on the linear space of adjoint-symmetries $Q_A$. 
Algebraic properties of these actions can be found in \Ref{AncWan2020b}.

\section{Geometrical adjoint-symmetries of evolution equations}\label{sec:evolpde}

A general system of evolution equations of order $N$ has the form 
\begin{equation}\label{pde.evol}
u_t^\alpha =g^\alpha(x,u,\partial_x u,\ldots,\partial_x^N u)
\end{equation}
where
$t$ is the time variable, 
$x^i$, $i=1,\ldots,n$, are now the space variables,
and $u^\alpha$, $\alpha=1,\ldots,m$, are the dependent variables. 
The space of formal solutions $u^\alpha(t,x)$ of the system will be denoted $\Esp$. 

The developments for general PDE systems can be specialized to evolution systems, 
with $G^\alpha = u_t^\alpha - g^\alpha$ 
via identifying the indices $A=\alpha$ ($M=m$).
On $\Esp$, 
since $u_t^\alpha$ can be eliminated through the evolution equations, 
the components of symmetries and adjoint-symmetries 
can be assumed to contain only $u^\alpha$ and its spatial derivatives
in addition to $t$ and $x^i$. 
Hereafter, multi-indices will refer to spatial derivatives. 

A symmetry is thereby an evolutionary vector field 
\begin{equation}\label{P.vector.evol}
\X_P = P^\alpha(t,x,\partial_x u,\ldots,\partial_x^k u)\partial_{u^\alpha}
\end{equation}
satisfying the linearization of the evolution system on $\Esp$:
\begin{equation}\label{symm.deteqn.evol}
(\pr\X_P(u_t^\alpha - g^\alpha))|_\Esp = (D_t P^\alpha - g'(P)^\alpha)|_\Esp =0 . 
\end{equation}
Off of $\Esp$, 
$D_t P^\alpha = (P_t + P'(g))^\alpha + P'(G)^\alpha$,
whereby $R_P= P'$. 
Consequently, the symmetry determining equation \eqref{symm.deteqn.evol}
can be expressed simply as 
\begin{equation}\label{symm.deteqn.evol.offsolnsp}
(P_t + [g,P])^\alpha =0 . 
\end{equation}

The determining equation for adjoint-symmetries $Q_\alpha(t,x,\partial_x u,\ldots,\partial_x^l u)$
is given by the adjoint linearization of the evolution system on $\Esp$:
\begin{equation}\label{adjsymm.deteqn.evol}
(-D_t Q - g'{}^*(Q))_\alpha|_\Esp =0 . 
\end{equation}
Similarly to the symmetry case, 
here $R_Q= -Q'$ off of $\Esp$,
and the adjoint-symmetry determining equation simply becomes
\begin{equation}\label{adjsymm.deteqn.evol.offsolnsp}
(Q_t + Q'(g) +g'{}^*(Q))_\alpha =0 . 
\end{equation}

These two determining equations have a geometrical formulation 
given by a Lie derivative defined in terms of a flow arising from the evolution system,
similar to the situation for ODEs \cite{Sar}.
Specifically, observe that $D_t u^\alpha |_\Esp = g^\alpha$,
and hence $D_t f|_\Esp = f_t + f'(g)$ for any differential function $f$. 
This motivates introducing the flow vector field 
\begin{equation}\label{flow.vector}
\Y = \partial_t + g^\alpha\partial_{u^\alpha}
\end{equation}
which is related to the total time derivative by prolongation, 
\begin{equation}\label{prolong.flow}
\pr\Y = D_t|_\Esp = \partial_t + (D_I g^\alpha)\partial_{u^\alpha_I} . 
\end{equation}

Associated to this flow vector field is the Lie derivative
\begin{equation}\label{lieder.flow}
\lieder{t}:= \lieder{\pr\Y}
\end{equation}
which acts on differential functions by $\lieder{t} f = \pr\Y(f) = D_t f|_\Esp$. 
On evolutionary vector fields \eqref{P.vector.evol}, 
this Lie derivative acts in the standard way as a commutator 
\begin{equation}\label{lieder.P.vector}
\begin{aligned}
\lieder{t} \pr \X_P 
& = \pr( (\pr\Y(P) - \pr\X_P(g))^\alpha \partial_{u^\alpha} )\\
& = \pr( (P_t + P'(g) - g'(P))^\alpha \partial_{u^\alpha} )\\
&= \pr( (P_t + [g,P])^\alpha \partial_{u^\alpha} ) . 
\end{aligned}
\end{equation}
Thus, the symmetry determining equation \eqref{symm.deteqn.evol.offsolnsp}
can be formulated as the vanishing of the Lie derivative expression \eqref{lieder.P.vector}. 
This establishes the following well-known geometrical result. 

\begin{proposition}\label{prop:symm.evol}
A symmetry of an evolution system \eqref{pde.evol}
is an evolutionary vector field \eqref{P.vector.evol} that is invariant under the associated flow \eqref{lieder.flow}. 
\end{proposition}

In particular, the resulting Lie-derivative vector field 
\begin{equation}\label{vector.lieder}
\lieder{t} \X_P = (P_t + [g,P])^\alpha \partial_{u^\alpha}
\end{equation}
vanishes iff the functions $P_\alpha$ are the components of a symmetry. 

A similar characterization will now be given for adjoint-symmetries, 
based on viewing the adjoint relation 
between the determining equations \eqref{symm.deteqn.evol.offsolnsp} and \eqref{adjsymm.deteqn.evol.offsolnsp}
as a duality relation between vectors and 1-forms. 

Introduce the evolutionary 1-form 
\begin{equation}\label{Q.1form.evol}
\evw_Q = Q_\alpha(t,x,\partial_x u,\ldots,\partial_x^l u) \d u^\alpha . 
\end{equation}
Its Lie derivative is given by 
\begin{equation}\label{lieder.Q.1form}
\begin{aligned}
\lieder{t} \evw_Q 
& = (\lieder{t} Q_\alpha) \d u^\alpha + Q_\alpha\lieder{t}(\d u^\alpha)\\
& = (Q_t +Q'(g))_\alpha \d u^\alpha + Q_\alpha \d(\lieder{t} u^\alpha)\\
& = (Q_t +Q'(g))_\alpha \d u^\alpha + Q_\alpha \d g^\alpha \\
& = (Q_t +Q'(g) + g'{}^*(Q))_\alpha \d u^\alpha  \text{ (mod total $D$)} . 
\end{aligned}
\end{equation}
This shows that the adjoint-symmetry determining equation \eqref{adjsymm.deteqn.evol.offsolnsp}
can be formulated as the functional vanishing of the Lie derivative expression \eqref{lieder.Q.1form}. 

\begin{theorem}\label{thm:adjsymm.evol}
An adjoint-symmetry of an evolution system \eqref{pde.evol}
is an evolutionary 1-form \eqref{Q.1form.evol} that is functionally invariant under the associated flow \eqref{lieder.flow}. 
\end{theorem}

In particular, the resulting Lie-derivative 1-form 
\begin{equation}\label{1form.lieder}
\lieder{t} \evw_Q  = (Q_t +Q'(g) + g'{}^*(Q))_\alpha \d u^\alpha  \text{ (mod total $D$)}
\end{equation}
functionally vanishes iff the functions $Q_\alpha$ are the components of an adjoint-symmetry. 
This 1-form \eqref{1form.lieder} is functionally equivalent to the adjoint-symmetry 1-form \eqref{Q.1form}
introduced for a general PDE system. 
To see the relationship in detail, observe that 
\begin{equation}\label{full1form.evol1form.rel}
\begin{aligned}
\w_Q = Q_\alpha \d G^\alpha 
& = Q_\alpha \d(u_t^\alpha -g^\alpha) \\
& = Q_\alpha ( D_t(\d u^\alpha) -g'(\d u)^\alpha ) \\
& = - (D_t Q_\alpha + g'{}^*(Q)_\alpha) \d u^\alpha \text{ (mod total $D$)} \\
& = -\lieder{t} \evw_Q \text{ (mod total $D$)} . 
\end{aligned}
\end{equation}

An interesting question is how to extend this relationship to more general PDE systems.

\subsection{Evolution equations with spatial constraints}

A wide generalization of evolution systems occurring in applied mathematics and mathematical physics 
is given by systems comprised of evolution equations with spatial constraints. 
Some notable examples are
Maxwell's equations,
incompressible fluid equations, 
magnetohydrodynamical equations, 
and Einstein's equations. 

The constraints in such systems in general consist of spatial equations 
\begin{equation}\label{pde.constr}
C^\Upsilon(x,u,\partial_x u,\ldots,\partial_x^{N'} u) =0,
\quad
\Upsilon = 1,\ldots,M'
\end{equation}
that are compatible with the evolution equations \eqref{pde.evol}. 
Compatibility means that the time derivative of the constraints vanishes 
on the solution space $\Esp$ of the whole system, 
$(D_t C^\Upsilon)|_\Esp = 0$. 
For systems that are regular \cite{Anc-review},
Hadamard's lemma implies that the system obeys a differential identity 
\begin{equation}\label{C.G.diffid}
D_t C^\Upsilon = C'(G)^\Upsilon + \Dop(C)^\Upsilon
\end{equation}
where $G^\alpha = u_t^\alpha -g^\alpha$ denotes the evolution equations \eqref{pde.evol},
and where $\Dop$ is a linear differential spatial operator 
whose coefficients are non-singular on $\Esp$. 
Equivalently, the constraints must obey the identity 
$C'(g)^\Upsilon = \Dop(C)^\Upsilon$. 
A comparison of the differential order of each side of this identity shows that 
$\Dop$ is of same order $N$ as the evolution equations, namely
\begin{equation}\label{Dop}
\Dop= \sum_{0\leq|I|\leq N} R^I{}_{\Lambda}^{\Upsilon} D_I . 
\end{equation}  

The full system consists of $n+M'$ equations $G^\alpha =0$, $C^\Upsilon=0$.
Note that, in the previous notation \eqref{pde},
$(G^\alpha,C^\Upsilon)=(G^A)$ with $A=(\alpha,\Upsilon)$.

The symmetry determining equation is given by 
the linearization of the full system on $\Esp$, 
which is comprised by the evolution part \eqref{symm.deteqn.evol} 
and the constraint part 
\begin{equation}\label{symm.deteqn.constr}
(\pr\X_P C^\Upsilon)|_\Esp = C'(P)^\Upsilon|_\Esp =0 . 
\end{equation}
Off of $\Esp$, 
$C'(P)^\Upsilon = R_C(C)^\Upsilon$,
where $R_C$ is a linear differential spatial operator 
whose coefficients are non-singular on $\Esp$. 
Hence, 
the determining equations \eqref{symm.deteqn.evol} and \eqref{symm.deteqn.constr} 
can be stated as 
\begin{equation}\label{symm.deteqn.evolconstr.offsolnsp}
(P_t + [g,P])^\alpha|_{\Esp_C} =0,
\quad
C'(P)^\Upsilon|_{\Esp_C} =0 
\end{equation}
where $\Esp_C$ denotes the solution space of the spatial constraint equations \eqref{pde.constr}. 

The adjoint-symmetry determining equation is given by 
the adjoint linearization of the full system on $\Esp$, 
which comprises evolution terms and additional constraint terms:
\begin{equation}\label{adjsymm.deteqn.evolconstr}
(-D_t Q - g'{}^*(Q) + C'{}^*(q) )_\alpha|_\Esp =0 . 
\end{equation}
Here the components of an adjoint-symmetry consist of 
\begin{equation}
(Q_\alpha(t,x,\partial_x u,\ldots,\partial_x^l u),q_\Upsilon(t,x,\partial_x u,\ldots,\partial_x^{l'} u))
\end{equation}
with $Q_\alpha$ being associated to the evolution equations as before,
while $q_\Upsilon$ is associated to the constraint equations. 
Similarly to the symmetry case, the determining equation can be stated as 
\begin{equation}\label{adjsymm.deteqn.evolconstr.offsolnsp}
(Q_t + Q'(g) +g'{}^*(Q) - C'{}^*(q) )_\alpha|_{\Esp_C} =0 . 
\end{equation}

These determining equations for symmetries and adjoint-symmetries 
have a geometrical formulation in terms of a constrained flow \eqref{flow.vector}, 
generalizing the previous formulation for evolution systems as follows. 

\begin{theorem}\label{thm:symm.evolconstr}
A symmetry of a constrained evolution system \eqref{pde.evol} and \eqref{pde.constr}
is an evolutionary vector field \eqref{P.vector.evol} 
that is invariant under the associated constrained flow \eqref{lieder.flow} 
and that preserves the constraints. 
\end{theorem} 

The proof of this result is simply the observation that, 
first, the determining equation \eqref{symm.deteqn.constr} 
corresponds to the constraints being preserved,
and second, the Lie derivative of the symmetry vector field \eqref{vector.lieder} 
along the flow vanishes on the constraint solution space. 

\begin{theorem}\label{thm:adjsymm.evolconstr}
An adjoint-symmetry of a constrained evolution system \eqref{pde.evol} and \eqref{pde.constr}
is an evolutionary 1-form \eqref{Q.1form.evol} that is functionally invariant 
under the associated constrained flow \eqref{lieder.flow}, 
up to a functional multiple of the normal 1-form $d C^\Upsilon$ 
arising from the constraints.
\end{theorem}

The proof is given by the earlier computation \eqref{1form.lieder} 
for the Lie derivative of the adjoint-symmetry 1-form. 
This computation shows that the adjoint-symmetry determining equation \eqref{adjsymm.deteqn.evolconstr.offsolnsp}
now can be expressed as 
\begin{equation}\label{1form.lieder.constr} 
\lieder{t} \evw_Q|_{\Esp_C}  
= (C'{}^*(q)_\alpha \d u^\alpha)|_{\Esp_C} 
= (q_\Upsilon \d C^\Upsilon)|_{\Esp_C} \text{ (mod total $D$)}
\end{equation}
where $\d C^\Upsilon$ is the normal 1-form given by the constraints viewed as surfaces in jet space. 

The Lie-derivative 1-form \eqref{1form.lieder.constr} 
is functionally equivalent to the adjoint-symmetry 1-form \eqref{Q.1form}
introduced for a general PDE system. 
In the present notation, 
the full system of evolution and constraint equations \eqref{pde.evol} and \eqref{pde.constr}
consists of $(G^\alpha,C^\Upsilon)=0$,
and the corresponding 1-form associated to this system is given by 
$\w_{Q,q} = Q_\alpha \d G^\alpha + q_\Upsilon \d C^\Upsilon$. 
Now observe that 
\begin{equation}
\w_{Q,q} = q_\Upsilon \d C^\Upsilon -\lieder{t} \evw_Q \text{ (mod total $D$)}
\end{equation}
using the relation \eqref{full1form.evol1form.rel}.

There is a class of adjoint-symmetries arising from
the summed product of arbitrary functions $\chi_\Upsilon(t,x)$
and the components of the the differential identity \eqref{C.G.diffid}.
This yields, after integration by parts, 
\begin{equation}
\begin{aligned}
0 & = \chi_\Upsilon (D_t C^\Upsilon - C'(G)^\Upsilon - \Dop(C)^\Upsilon)\\
& = D_t (\chi_\Upsilon C^\Upsilon) + D_i \Psi^i(\chi,G;C) - D_i\Phi^i(\chi,C;R)
- (D_t\chi +\Dop^*(\chi))_\Upsilon C^\Upsilon - C'{}^*(\chi)_\alpha G^\alpha
\end{aligned}
\end{equation}
where
$\Phi^i(\chi,C;R) =\sum_{0\leq|I|\leq N-1} (-1)^{|J|} D_J(\chi_\Upsilon R^{iI}{}_{\Lambda}^{\Upsilon}) D_{I/J} C^\Lambda$
from expression \eqref{Dop}. 
Hence,
\begin{equation}
D_t (\chi_\Upsilon C^\Upsilon) + D_i (\Psi^i(\chi,G;C) - \Phi^i(\chi,C;R))
= C'{}^*(\chi)_\alpha G^\alpha + (D_t\chi+\Dop^*(\chi))_\Upsilon C^\Upsilon
\end{equation}
has the form of a conservation law off $\Esp$,
with $(C'{}^*(\chi)_\alpha,(D_t\chi +\Dop^*(\chi))_\Upsilon)$ being the multiplier.
As is well known,
every multiplier for a regular PDE system is an adjoint-symmetry
\cite{Vin1984,Olv-book,BCA-book,Anc-review}.
This can be proven here by applying the Euler operator $E_{u^\alpha}$
and using its product rule.
Consequently, 
\begin{equation}\label{gauge.adjsymm.evolconstr}
Q_\alpha = C'{}^*(\chi)_\alpha,
\quad
q_\Upsilon = (D_t\chi+\Dop^*(\chi))_\Upsilon
\end{equation}
are components of an adjoint-symmetry,
involving the arbitrary functions $\chi_\Upsilon(t,x)$.
Such adjoint-symmetries are a counterpart of gauge symmetries,
and accordingly are called \emph{gauge adjoint-symmetries} \cite{Anc-review}.

The corresponding gauge adjoint-symmetry 1-form is given by
\begin{equation}
\evw_\chi = C'{}^*(\chi)_\alpha \d u^\alpha = \chi_\Upsilon \d C^\Upsilon \text{ (mod total $D$)}
\end{equation}
and satisfies the geometrical relation
\begin{equation}
\lieder{t} \evw_\chi|_{\Esp_C}  
= ( (D_t\chi+\Dop^*(\chi))_\Upsilon \d C^\Upsilon)|_{\Esp_C} \text{ (mod total $D$)} . 
\end{equation}
This establishes the following geometrical result. 

\begin{theorem}\label{thm:gauge.adjsymm}
A gauge adjoint-symmetry \eqref{gauge.adjsymm.evolconstr}
is functionally equivalent to a normal 1-form $\evw_\chi$ associated to the constraint equations \eqref{pde.constr}.
Under the evolution flow, it is mapped into another normal 1-form.
\end{theorem}

The preceding developments for general systems of evolution equations with spatial constraints
have used the classical notion of symmetries and adjoint-symmetries. 
It would be interesting to extend the formulation and the results 
by considering a notion of conditional symmetries and corresponding conditional adjoint-symmetries
based on the spatial constraints. 

Specifically, on the solution space of the full system, 
consider a symmetry given by an evolutionary vector field \eqref{P.vector.evol} 
that satisfies 
\begin{equation}\label{symm.deteqn.evolconstr.conditional}
(P_t + [g,P])^\alpha|_{\Esp_C} =0
\end{equation}
where $\Esp_C$ denotes the solution space of the spatial constraint equations \eqref{pde.constr}. 
Such conditional symmetries \eqref{symm.deteqn.evolconstr.conditional} 
differ from classical symmetries \eqref{symm.deteqn.evolconstr.offsolnsp}
by relaxing the condition that the constraints are preserved. 
Their natural adjoint counterpart is given by an evolutionary 1-form \eqref{Q.1form.evol} satisfying
\begin{equation}\label{adjsymm.deteqn.evolconstr.conditional}
(Q_t + Q'(g) + g'{}^*(Q))_\alpha|_{\Esp_C} =0 . 
\end{equation}
which is the adjoint of the determining equation \eqref{symm.deteqn.evolconstr.conditional}. 
Such conditional adjoint-symmetries \eqref{adjsymm.deteqn.evolconstr.conditional}
differ from classical adjoint-symmetries \eqref{adjsymm.deteqn.evolconstr.offsolnsp}
by excluding the terms arising from the spatial constraints.

This notion of conditional symmetries and adjoint-symmetries is more general than
the classical notion because the conditional determining equations 
hold on $\Esp_C$ instead of the whole jet space.

\section{Concluding remarks}\label{sec:remarks}

The main results showing how adjoint-symmetries correspond to evolutionary 1-forms with certain geometrical properties
provides a first step towards giving a fully geometrical interpretation for adjoint-symmetries.
In particular, for systems of evolution equations, 
adjoint-symmetries can be geometrically described as 1-forms that are invariant under the flow generated by the system on the solution space.
This interesting result has a straightforward generalization to systems of evolution equations with spatial constraints.
Consequently, the results presented here are applicable to all PDE systems of
interest in applied mathematics and mathematical physics. 

One direction for future work will be to translate and generalize
these results into the abstract geometrical setting of secondary calculus
\cite{KraVin-book,Vin1998}
developed by Vinogradov and Krasil'shchik and their co-workers. 

It will also be interesting to develop fully the use of adjoint-symmetries 
in the study of specific PDE systems, as outlined in the introduction: 
finding exact solutions, 
detecting and finding mappings into a target class of PDEs,
and detecting integrability,
which are counterparts of some important uses of symmetries. 
Another use of adjoint-symmetries, which has been introduced very recently \cite{AncWan2020a}, 
is for finding pre-symplectic operators.

\end{document}